\documentclass[a4paper]{article}
\usepackage{a4wide}

\usepackage{adjustbox}
\usepackage[cref,theorems,pdf]{nikis}
\usepackage{url}
\usepackage{enumitem}
\newcommand*{\teigi}[1]{\emph{#1}}

\usepackage{graphicx}

\newcommand*{\LyndonTree}{\mathit{LTree}}
\newcommand*{\derive}{\mathit{val}}
\newcommand*{\Expr}[1]{\ensuremath{\mathrm{expr}_{#1}}}
\newcommand*{\lfs}[2]{\ensuremath{\mathit{lfs}_{#1}(#2)}}
\newcommand*{\SLP}{\ensuremath{\mathcal{G}_{\mathrm{SLP}}}}
\newcommand*{\CFG}{\ensuremath{\mathcal{G}_{\mathrm{AG}}}}
\newcommand*{\LYN}{\ensuremath{\mathcal{G}_{\mathrm{LYN}}}}

\newcommand*{\LeftRule}[1][i]{\ensuremath{X_{{#1}_{\mathrm{L}}}}}
\newcommand*{\RightRule}[1][i]{\ensuremath{X_{{#1}_{\mathrm{R}}}}}
\newcommand*{\LeftPattern}{\ensuremath{P_{\mathrm{L}}}}
\newcommand*{\RightPattern}{\ensuremath{P_{\mathrm{R}}}}

\newcommand*{\RunningExample}{aababaababb}
\newcommand*{\Locate}{\ensuremath{\mathit{locate}}}
\newcommand*{\occ}[1][]{\mathrm{occ}}
\newcommand*{\ag}[1][]{\hat{g}}

\newcommand*{\myCite}[2]{\cite[#1]{#2}}
\newcommand*{\Cite}{\cite}
\newcommand*{\Tabref}{Table~\ref}
\newcommand*{\Figref}{Figure~\ref}
\newcommand*{\figref}{Figure~\ref}

\begin{document}%

\title{Grammar-compressed Self-index with \\ Lyndon Words}

\newcommand{\inst}[1]{\ensuremath{^{#1}}}
\author{
  Kazuya Tsuruta\inst{1} \and
  Dominik K\"{o}ppl\inst{1,2} \and
  Yuto Nakashima\inst{1} \and
  Shunsuke Inenaga\inst{1,3} \and
  Hideo Bannai\inst{1} \and
  Masayuki Takeda\inst{1} \and
  {
    \small
    \begin{minipage}{\linewidth}
      \begin{center}
        $^1$Department of Informatics, Kyushu University, Fukuoka, Japan,
        \\ \texttt{\{kazuya.tsuruta, yuto.nakashima, inenaga, bannai, takeda\}@inf.kyushu-u.ac.jp}
        \\
        $^2$Japan Society for Promotion of Science, \texttt{dominik.koeppl@inf.kyushu-u.ac.jp}
        \\
        $^3$Japan Science and Technology Agency
      \end{center}
    \end{minipage}
  }
}
\date{}

\maketitle

\begin{abstract}
  We introduce a new class of straight-line programs (SLPs),
named the \emph{Lyndon SLP},
inspired by the Lyndon trees~(Barcelo, 1990).
Based on this SLP,
we propose a self-index data structure of $\Oh{g}$ words of space
that can be built from a string~$T$ in $\Oh{n \lg n}$ expected time,
retrieving the starting positions of all occurrences of a pattern~$P$ of length~$m$ in $\Oh{m + \lg m \lg n + \occ \lg g}$ time,
where $n$ is the length of~$T$,
$g$ is the size of the Lyndon SLP for $T$,
and $\occ$ is the number of occurrences of $P$ in $T$.

\end{abstract}

\section{Introduction}

A context-free grammar is said to \emph{represent} a string $T$
if it generates the language consisting of $T$ and only $T$.
Grammar-based compression~\cite{DBLP:journals/tit/KiefferY00} is, given a string $T$,
to find a small size description of $T$ based on a context-free grammar that represents $T$.
The grammar-based compression scheme is known to be most suitable for compressing
\emph{highly-repetitive strings}.
Due to its ease of manipulation, grammar-based representation of strings is a frequently used model for \emph{compressed string processing},
where the aim is to efficiently process compressed strings without explicit decompression.
Such an approach allows for theoretical and practical speed-ups compared to a naive decompress-then-process approach.

A \emph{self-index} is a data structure that
is a full-text index, i.e., supports various pattern matching queries on the text,
and also provides random access to the text,
usually without explicitly holding the text itself.
Examples are
the compressed suffix array~\cite{DBLP:conf/stoc/GrossiV00,DBLP:conf/isaac/HonLSS03,DBLP:conf/cocoon/LamSSY02},
the compressed compact suffix array~\cite{DBLP:conf/cpm/MakinenN04}, and
the FM index~\cite{DBLP:conf/focs/FerraginaM00}.%
\footnote{Navarro and M{\"{a}}kinen~\cite{DBLP:journals/csur/NavarroM07} published an excellent survey on this topic.}
These self-indexes are, however, unable to fully exploit the redundancy of highly repetitive strings.
To exploit such repetitiveness, Claude and Navarro~\cite{claudear:_self_index_gramm_based_compr}
proposed the first self-index based on grammar-based compression.
The method is based on a \emph{straight-line program (SLP)},
a context-free grammar representing a single string in the Chomsky normal form.
Plenty of grammar-based self-indexes have already been proposed (e.g.,~\cite{DBLP:conf/spire/ClaudeN12a,DBLP:conf/wea/TakabatakeTS14,takabatake16siedm,NISHIMOTO2019}).

In this paper, we first introduce a new class of SLPs,
named the \emph{Lyndon SLP},
inspired by the Lyndon tree~\cite{barcelo90:_free_lie_algeb}.
We then propose a self-index structure of $\Oh{g}$ words of space
that can be built from a string~$T$ in $\Oh{n \lg n}$ expected time.
The proposed self-index can find the starting positions of all occurrences of a pattern~$P$ of length~$m$ in
$\Oh{m + \lg m \lg n + \occ \lg g}$ time,
where $n$ is the length of~$T$,
$g$ is the size of the Lyndon SLP for $T$,
$\sigma$ is the alphabet size, $w$ is the computer word length
and $\occ$ is the number of occurrences of $P$ in $T$.

\subsection{Related work}

The \emph{smallest grammar problem} is, given a string $T$,
to find the context-free grammar $G$ representing $T$ with the smallest possible size,
where the \emph{size} of $G$ is the total length
of the right-hand sides of the production rules in $G$.
Since the smallest grammar problem is NP-hard~\cite{Storer77},
many attempts have been made to develop small-sized context-free grammars representing a given string~$T$.
LZ78~\cite{LZ78},
LZW~\cite{LZW},
Sequitur~\cite{SEQUITUR},
Sequential~\cite{DBLP:journals/tit/KiefferY00},
LongestMatch~\cite{DBLP:journals/tit/KiefferY00},
Re-Pair~\cite{LarssonDCC99},
and
Bisection~\cite{Kieffer00MPM} are
grammars based on simple greedy heuristics.
Among them Re-Pair is known for achieving high compression ratios in practice.

Approximations for the smallest grammar have also been proposed.
The AVL grammars~\cite{rytter_tcs03} and
the $\alpha$-balanced grammars~\cite{charikar05:_small_gramm_probl}
can be computed in linear time and achieve
the currently best approximation ratio of $\Oh{\lg(|T|/g_T^*)}$ by using the LZ77 factorization and the balanced binary grammars,
where $g_T^*$ denotes the smallest grammar size for $T$.
Other grammars with linear-time algorithms achieving the approximation $\Oh{\lg(|T|/g_T^*)}$
are LevelwiseRePair~\cite{sakamoto05} and
Recompression%
~\cite{jez15recompression}.
They basically replace di-grams with a new variable in a bottom-up manner similar to Re-Pair, but use different mechanisms to select the di-grams.
On the other hand,
LCA~\cite{sakamoto04:_space_savin_linear_time_algor} and
its variants~\cite{sakamoto09:_space_savin_approx_algor_gramm_based_compr,maruyama12:_onlin_algor_light_gramm_based_compr,maruyama13:_fully_onlin_gramm_compr}
are known as scalable practical approximation algorithms.
The core idea of LCA is the \emph{edit-sensitive parsing (ESP)}~\cite{DBLP:journals/talg/CormodeM07}, a parsing algorithm developed for approximately computing the edit distance with moves.
The \emph{locally-consistent-parsing (LCP)}~\cite{sahinalp95:_data} is a generalization of ESP\@.
The \emph{signature encoding (SE)}~\cite{DBLP:journals/algorithmica/MehlhornSU97},
developed for equality testing on a dynamic set of strings,
is based on LCP and can be used as a grammar-transform method.
The ESP index~\cite{DBLP:conf/wea/TakabatakeTS14,takabatake16siedm} and
the SE index~\cite{NISHIMOTO2019}
are grammar-based self-indexes based on ESP and SE, respectively.

While our experimental section (\cref{secExperiments}) serves as a practical comparison between the sizes of some of the above mentioned grammars,
\Tabref{tabConstruction} gives a comparison with some theoretically appealing index data structures based on grammar compression.
There, we chose the indexes of Claude and Navarro~\Cite{DBLP:conf/spire/ClaudeN12a}, Gagie et al.~\Cite{DBLP:conf/lata/GagieGKNP12}, and Christiansen et al.~\Cite{christiansen18optimaltime}
because these indexes have non-trivial time complexities for answering queries.
We observe that our proposed index has the fastest construction among the chosen grammar indexes,
and is competitively small if $g = \oh{\gamma \lg (n/\gamma)}$ while being clearly faster than the first two approaches for long patterns.
It is worth pointing out that the grammar index of Christiansen et al.~\cite{christiansen18optimaltime} achieves a grammar size whose upper bound \Oh{\gamma \lg (n/\gamma)} matches the upper bound of the size~$g_T^*$ of the smallest possible grammar.
Unfortunately, we do not know how to compare our result within these terms in general.

\begin{table*}
    \caption{Complexity bounds of self-indexes based on grammar compression.}
    \label{tabConstruction}
    \begin{center}
        \begin{tabular}{l|*{2}{l}}
            \multicolumn{3}{c}{Construction space (in words) and time} \\
            \hline \hline
            \multicolumn{1}{c|}{Index} &
            \multicolumn{1}{c}{Space} &
            \multicolumn{1}{c}{Time} \\
            \hline
            \Cite{DBLP:conf/spire/ClaudeN12a}
            & ${\Oh{n}}$ & ${\Oh{n + \ag\lg\ag}}$ \\
            \Cite{christiansen18optimaltime}
            & ${\Oh{n}}$ & ${\Oh{n\lg n}} $ expected \\
            \Cite{christiansen18optimaltime}
            & ${\Oh{n}}$ & ${\Oh{n \lg n}} $ expected \\
            This paper
            & ${\Oh{n}}$ & ${\Oh{n \lg n}} $ expected\\
            \hline
        \end{tabular}
        \begin{tabular}{l|*{2}{l}}
            \multicolumn{3}{c}{Needed space (in words) and query time for a pattern of length~$m$} \\
            \hline \hline
            \multicolumn{1}{c|}{Index} &
            \multicolumn{1}{c}{Space} &
            \multicolumn{1}{c}{Locate Time} \\
            \hline
            \Cite{DBLP:conf/spire/ClaudeN12a}
            & ${\Oh{\ag}}$ & ${\Oh{m^2 \lg \lg_{\ag} n + (m + \occ)\lg\ag}}$ \\
            \Cite{DBLP:conf/lata/GagieGKNP12}
            & ${\Oh{\ag + z \lg \lg z}}$ & ${\Oh{m^2 + (m + \occ) \lg\lg n}}$ \\
            \Cite{christiansen18optimaltime}
            & ${\Oh{\gamma\lg(n/\gamma)}}$ & ${\Oh{m + \lg^\epsilon \gamma + \occ \lg^\epsilon(\gamma\lg(n/\gamma))}}$ \\
            \Cite{christiansen18optimaltime}
            & ${\Oh{\gamma\lg(n/\gamma)\lg^\epsilon(\gamma\lg(n/\gamma))}}$ & ${\Oh{m + \occ}}$ \\
            Theorem \ref{th:collision_free}
            & ${\Oh{g}}$ & ${\Oh{m + \lg m \lg n + \occ \lg g}}$ \\
            \hline
        \end{tabular}
    \end{center}
    {\small
        $n$ is the length of~$T$,
        $z$ is the number of LZ77~\cite{ziv77lz} phrases of~$T$,
        $\gamma$ is the size of the smallest string attractor~\cite{kempa18stringattractors} of~$T$,
        $g$ is the size of the Lyndon SLP of~$T$,
        $\ag$ is the size of a given admissible grammar,
        $\epsilon > 0$ is a constant,
        $m$ is the length of a pattern~$P$,
        and $\occ$ is the number of occurrences of $P$ in $T$.
    }
\end{table*}

\section{Preliminaries}
\label{sec:preliminaries}

\subsection{Notation}
Let $\Sigma$ be an ordered finite {\em alphabet}.
An element of $\Sigma^*$ is called a {\em string}.
The length of a string $S$ is denoted by $|S|$.
The empty string $\varepsilon$ is the string of length 0.
For a string $S = xyz$, $x$, $y$ and $z$ are called
a \emph{prefix}, \emph{substring}, and \emph{suffix} of $S$, respectively.
A prefix (resp.\ suffix) $x$ of $S$ is called a \emph{proper prefix} (resp.\ suffix)
of $S$ if $x \neq S$.
$S^\ell$ denotes the $\ell$ times concatenation of the string~$S$.
The $i$-th character of a string $S$ is denoted by $S[i]$, where $i \in [1..|S|]$.
For a string $S$ and two integers $i$ and $j$ with $1 \leq i \leq j \leq |S|$,
let $S[i..j]$ denote the substring of $S$ that begins at position $i$ and ends at
position $j$. For convenience, let $S[i..j] = \varepsilon$ when $i > j$.

\subsection{Lyndon words and Lyndon trees}

Let $\preceq$ denote some total order on $\Sigma$ that induces
the lexicographic order~$\preceq$ on $\Sigma^*$.
We write $u \prec v$ to imply $u\preceq v$ and $u\neq v$
for any $u,v\in\Sigma^*$.

\begin{definition}[Lyndon Word~\cite{lyndon54}]
  A non-empty string $w \in \Sigma^+$ is said to be a \emph{Lyndon word}
  with respect to $\prec$ if $w \prec u$ for every non-empty proper suffix $u$ of $w$.
\end{definition}
By this definition, all characters $(\in \Sigma^1)$ are Lyndon words.

\begin{definition}[Standard Factorization~\cite{ChenFL58:_lyndon_factorization_,Lothaire83}]
  The \emph{standard factorization} of a Lyndon word $w$ with $|w|\geq 2$
  is an ordered pair $(u, v)$ of strings $u,v$
  such that $w = uv$ and
  $v$ is the longest proper suffix of $w$ that is also a Lyndon word.
\end{definition}

\begin{lemma}[\Cite{BassinoCN05,Lothaire83}]
For a Lyndon word $w$ with $|w| > 1$,
the standard factorization $(u, v)$ of $w$ always exists,
and the strings $u$ and $v$ are Lyndon words.
\end{lemma}

The Lyndon tree of a Lyndon word $w$, defined below,
is the full binary tree induced by recursively applying the standard factorization on $w$.

\begin{definition}[Lyndon Tree~\cite{barcelo90:_free_lie_algeb}]
  The \teigi{Lyndon tree} of a Lyndon word $w$, denoted by $\LyndonTree(w)$, is
  an ordered full binary tree defined recursively as follows:
  \begin{itemize}
  \item if $|w| = 1$, then $\LyndonTree(w)$ consists of a single node labeled by $w$;
  \item if $|w| \geq 2$, then the root of $\LyndonTree(w)$, labeled by $w$,
    has the left child $\LyndonTree(u)$ and the right child $\LyndonTree(v)$,
    where $(u, v)$ is the standard factorization of $w$.
  \end{itemize}
\end{definition}
\Figref{fig:lyndonTree} shows an example of a Lyndon tree for the
Lyndon word $\mathtt{aababaababb}$.

\begin{figure}[h]
  \centerline{
    \includegraphics[width=0.98\textwidth]{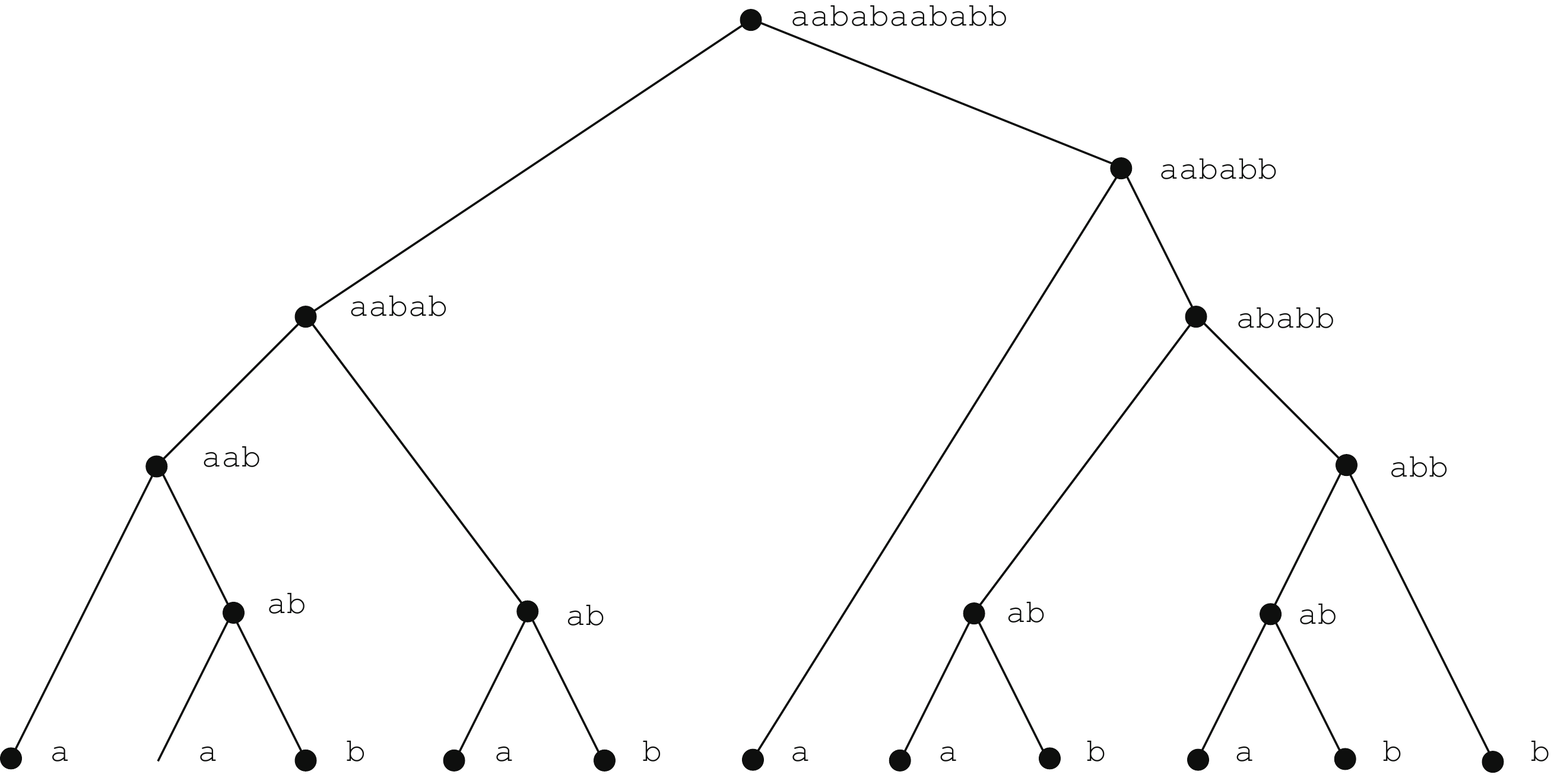}
  }
  \caption{
    The Lyndon tree for the Lyndon word $\mathtt{aababaababb}$ with respect to the
    order $\mathtt{a} \prec \mathtt{b}$,
    where each node is accompanied by its derived string to its right.
  }
  \label{fig:lyndonTree}
\end{figure}

\subsection{Admissible grammars and straight-line programs (SLPs)}
An admissible grammar~\cite{DBLP:journals/tit/KiefferY00} is a context-free grammar that generates a language consisting only of a single string.
Formally, an \teigi{admissible grammar (AG)} is a set of production rules $\CFG = \{ X_i \rightarrow \Expr{i} \}_{i = 1}^{g}$,
where $X_i$ is a \teigi{variable} and $\Expr{i}$ is a non-empty string over $\Sigma\cup\{X_1,\ldots,X_{i-1}\}$, called an \teigi{expression}.
The variable $X_g$ is called the \emph{start symbol}.
We denote by $\derive(X_i)$ the string derived by $X_i$.
We say that an admissible grammar \CFG{} \emph{represents} a string $T$ if $T = \derive(X_g)$.
To ease notation, we sometimes associate $\derive(X_i)$ with $X_i$.
The \teigi{size} of \CFG{} is the total length of all expressions~$\Expr{i}$.
We assume that any admissible grammar has no useless symbols.

It should be stated that
the above definition of admissible grammar is different with but equivalent to
the original definition in \Cite{DBLP:journals/tit/KiefferY00},
which defines an admissible grammar to be
a context-free grammar $G$ satisfying the conditions:
(1) $G$ is deterministic, i.e., for every variable $A$
there is exactly one production rule of the form $A \to \gamma$,
where $\gamma$ is a non-empty string consisting of variables and characters;
(2) $G$ has no production rule of the form $A \to \varepsilon$;
(3) The language $L(G)$ of $G$ is not empty; and
(4) $G$ has no useless symbols, i.e., every symbol appears in some derivation
that begins with the start symbol and
ends with a string consisting only of characters.

A \teigi{straight-line program} (\teigi{SLP}) is an admissible grammar
in the Chomsky normal form, namely,
each production rule is either of the form
$X_i \rightarrow a$ for some $a\in \Sigma$ or
$X_i \rightarrow \LeftRule \RightRule$ with $i > i_{\mathrm{L}}, i_{\mathrm{R}}$.
Note that $\SLP{}$ can derive a string up to length~$\Oh{2^g}$.
This can be seen by the example string $T = \mathtt{a} \cdots \mathtt{a}$ consisting of $n = 2^\ell$ \texttt{a}'s,
where the smallest SLP
$\{X_1 \rightarrow \mathtt{a} \} \cup \bigcup_{j=2}^{\ell + 1} \{ X_j \rightarrow X_{j-1} X_{j-1} \}$ has size $2\ell + 1$.

The derivation tree $\mathcal{T}_{\SLP{}}$ of \SLP{} is a labeled ordered binary tree,
where each internal node is labeled with a variable in $\{X_1, \ldots, X_g\}$,
and each leaf is labeled with a character in $\Sigma$.
The root node has the start symbol $X_g$ as label.
An example of the derivation tree of an SLP is shown in \figref{fig:LyndonSLP}.

\subsection{Grammar irreducibility}\label{secGrammarIrreducibility}
An admissible grammar is said to be \emph{irreducible} if it satisfies
the following conditions:
\begin{enumerate}[leftmargin=6ex, label=C-\arabic*.]
\item Every variable other than the start symbol is used more than once (\textbf{rule utility});
\item All pairs of symbols have at most one non-overlapping occurrence
  in the right-hand sides of the production rules (\textbf{di-gram uniqueness}); and
\item Distinct variables derive different strings.
\end{enumerate}
Grammar-based compression is a combination of
\begin{enumerate}
	\item the \emph{grammar transform}, i.e., the computation of 
	an admissible grammar $G$ representing the input string $T$, and
	\item the \emph{grammar encoding}, i.e., an encoding for $G$.
\end{enumerate}
Kieffer and Yang~\cite{DBLP:journals/tit/KiefferY00} showed that
a combination of an irreducible grammar transform
and a zero order arithmetic code is universal,
where a grammar transform is said to be \emph{irreducible} if
the resulting grammars are irreducible.

If an admissible grammar $G$ is not irreducible,
we can apply at least one of the following reduction rules~\cite{DBLP:journals/tit/KiefferY00} to make~$G$ irreducible:

\begin{enumerate}[leftmargin=6ex, label=R-\arabic*.]
\item 
	Replace each variable $X_i$ occurring only once in the right-hand sides of the production rules
	with $\Expr{i}$ and remove the rule $X_i\rightarrow \Expr{i}$.
  We also remove all production rules with useless symbols.
\item
	Given there are at least two non-overlapping occurrences of a string~$\gamma$ of symbols
  with $|\gamma|\ge 2$ in the right-hand sides of the production rules,
  replace each of the occurrences of $\gamma$ with the variable $X_i$, where
  $X_i \rightarrow \gamma$ is an existing or newly created production rule.
  Recurse until no such $\gamma$ longer exists.
\item
  For each two distinct variables $X_i$ and $X_j$ deriving an identical string,
  (a) replace all occurrences of $X_j$ with $X_i$ in the right-hand sides of the production rules,
  and (b) remove the production rule $X_j \rightarrow \Expr{j}$ and discard the variable~$X_j$.
  Consequently, there are no two distinct variables $X_i$ and $X_j$ with $\derive(X_i) = \derive(X_j)$.
  This operation possibly makes some variables useless;
  the production rules with such variables will be removed by R-1.
\end{enumerate}

\section{Lyndon SLP}

In what follows, we propose a new SLP, called Lyndon SLP\@.
A \teigi{Lyndon SLP} is an SLP $\LYN = \{ X_i \rightarrow \Expr{i} \}_{i = 1}^{g}$ representing a Lyndon word,
and satisfies the following properties:
  \begin{itemize}
  \item The strings $\derive(X_i)$ are Lyndon words for all variables $X_i$.
  \item The standard factorization of the string $\derive(X_i)$
    is $(\derive(\LeftRule),\derive(\RightRule))$
    for every rule $X_i \rightarrow \LeftRule \RightRule$.
  \item No pair of distinct variables $X_i$ and $X_j$
    satisfies $\derive(X_i) = \derive(X_j)$.
  \end{itemize}
  The derivation tree (when excluding its leaves) of $\mathcal{T}_{\LYN}$ is isomorphic to the Lyndon tree of~$T$ (cf.~\figref{fig:LyndonSLP}).

\begin{figure}
	\adjustbox{valign=c}{%
  \includegraphics[width=.7\textwidth]{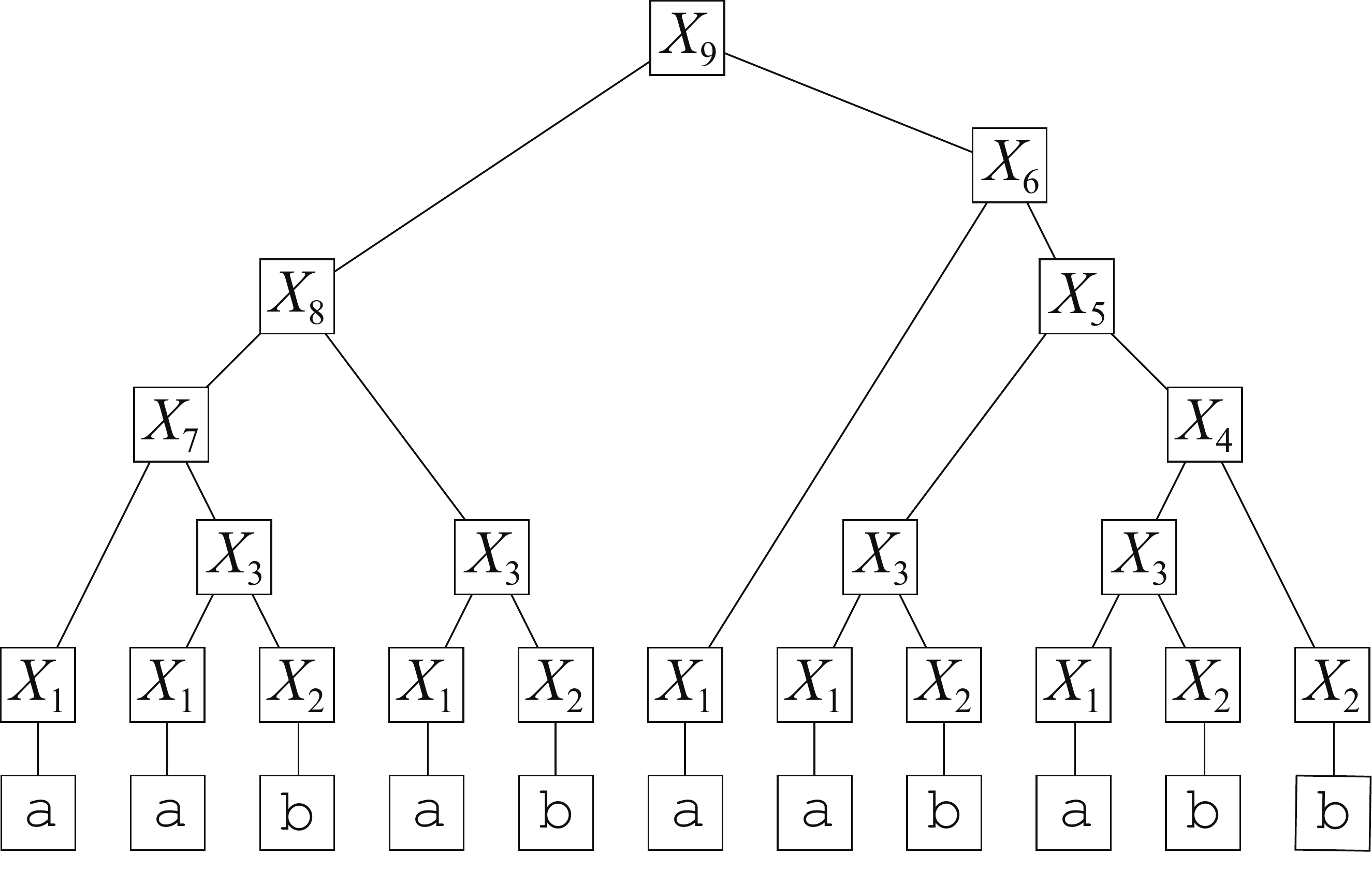}
	}%
	\hfill
	\adjustbox{valign=c}{%
		\(
  		\begin{array}{lll}
			X_1 & \rightarrow & \mathtt{a}\\
			X_2 & \rightarrow & \mathtt{b}\\
			X_3 & \rightarrow & X_1 X_2\\
			X_4 & \rightarrow & X_3 X_2\\
			X_5 & \rightarrow & X_3 X_4\\
			X_6 & \rightarrow & X_1 X_5\\
			X_7 & \rightarrow & X_1 X_3\\
			X_8 & \rightarrow & X_7 X_3\\
			X_9 & \rightarrow & X_8 X_6\\
  		\end{array}
		\)
	}%

  \caption{%
	  \emph{Left}: The derivation tree of the Lyndon SLP \LYN{} with $g = 9$
    representing the Lyndon word $T = \texttt{\RunningExample}$.
	\emph{Right}: The production rules of \LYN{}.
}
\label{fig:LyndonSLP}
\end{figure}

The rest of this article is devoted to algorithmic aspects regarding the Lyndon SLP\@.
We study its construction (\cref{secConstruction}), practically evaluate its size (\cref{secExperiments}), and propose an index data structure on it (\cref{secSelfIndex}).
For that, we work in the word RAM model supporting packing characters of sufficiently small bit widths into a single machine word.
Let~$w$ denote the machine word size in bits.

We fix a text $T[1..n]$ over an integer alphabet $\Sigma$ with size~$\sigma = n^{\Oh{1}}$.
If $T$ is not a Lyndon word, we prepend $T$ with a character smaller than all other characters appearing in~$T$.
Let $g$ denote the size $|\LYN{}|$ of the Lyndon SLP $\LYN{}$ of~$T$.

\begin{lemma}[{\myCite{Algo.~1}{bannai17runs}}]
We can construct the Lyndon tree of $T$ in \Oh{n} time.
\end{lemma}

\subsection{Constructing Lyndon SLPs}\label{secConstruction}
The algorithm of Bannai et al.~\cite[Algo.~1]{bannai17runs} builds the Lyndon tree \emph{online} from right to left.
We can modify this algorithm to create the Lyndon SLP of~$T$ by storing a dictionary for the rules and a reverse dictionary for looking up rules:
Whenever the algorithm creates a new node~$u$,
we query the reverse dictionary with $u$'s two children~$v$ and~$w$ for an existing rule~$X \rightarrow X_v X_w$, where $X_v$ and $X_w$ are the variables representing $v$ and $w$.
If such a rule exists, we assign $u$ the variable $X$, otherwise we create a new rule $X_u \rightarrow X_v X_w$ and put this new rule into both dictionaries.
The dictionaries can be implemented as balanced search trees or hash tables, featuring $\Oh{n \lg g}$ deterministic construction time or $\Oh{n}$ expected construction time, respectively.

In the static setting (i.e., we do not work online),
deterministic $\Oh{n}$ time can be achieved with the enhanced suffix array~\cite{manber93sa,abouelhoda04enhanced} supporting constant time longest common extension queries.
We associate each node~$v$ of the Lyndon tree with the pair $(|T[i..j]|, \mathit{rank}(i))$,
where $T[i..j]$ is the substring derived from the non-terminal representing $v$, and 
$\mathit{rank}(i)$ is the lexicographic rank of the suffix starting at position $i$.
Then, sort all nodes according to their associated pairs with a linear-time integer sorting algorithm.
By using longest common extension queries between adjacent nodes of equal length
in the sorted order,
we can determine in $\Oh{1}$ time per node whether they represent the same string,
and if so, assign the same variable (otherwise assign a new variable).

\subsection{Lyndon array simulation}\label{secLyndonArray}
As a by-product, we can equip the Lyndon SLP of~$T$ with the indexing data structure of Bille et al.~\cite{bille15randomaccess} to support character extraction and navigation in \Oh{\lg n} time.
This allows us to compute the $i$-th entry of the Lyndon array~\cite{bannai17runs} in \Oh{\lg n} time,
where the $i$-th entry of the Lyndon array of~$T$ stores the length of the longest Lyndon word starting at $T[i]$.
For that, given a text position~$i$, we search for the highest Lyndon tree node having $T[i]$ as its leftmost leaf.
Given the rightmost leaf of this node represents~$T[j]$, the longest Lyndon word starting at $T[i]$ has the length $j-i+1$.
(Otherwise, there would be a higher node in the Lyndon tree representing a longer Lyndon word starting at~$T[i]$.)

\begin{lemma}
	There is a data structure of size \Oh{g} that can retrieve the longest Lyndon word starting at $T[i]$ in \Oh{\lg n} time.
\end{lemma}

\subsection{Computational experiments}\label{secExperiments}

We empirically benchmark the grammar sizes obtained by the Lyndon SLP to highlight its potential as a grammar compressor.
As benchmark datasets we used four highly repetitive texts consisting of
the files
\texttt{cere},
\texttt{einstein.de.txt},
\texttt{kernel}, and
\texttt{world\_leaders}
from the Pizza \& Chili corpus
(\url{http://pizzachili.dcc.uchile.cl}).
We used the natural order implied by the ASCII code for building the Lyndon SLPs.
We compared the size of the resulting Lyndon grammars
with the resulting grammars of Re-Pair, LCA, Recompression.
We used existing implementations of Re-Pair
(\url{https://users.dcc.uchile.cl/~gnavarro/software/}) and of
LCA (\url{http://code.google.com/p/lcacomp/}).
The outputs of LCA, Recompression and our method are SLPs,
while those of Re-Pair are AGs (and not necessarily SLPs).
For a fair comparison, we compared the resulting grammar sizes either in an SLP representation, or in a common AG representation.

\begin{description}
\item[SLP]
We keep the resulting grammar of the Lyndon SLP, LCA, and Recompression as it is, but transform the output of Re-Pair to an SLP\@.
To this end, we observe that Re-Pair consists of (a) a list of non-terminals whose right hand sides are already of length two, and (b) a start symbol
whose right hand side is a string of symbols of arbitrary size. 
Consequently, to transform this grammar to an SLP, it is left to focus on the start symbol: 
We replace greedily di-grams in the right hand side of the start symbol until it consists only of two symbols.
\item[AG]
We process each grammar in the following way:
	First, we remove the production rules of the form $X_i\rightarrow a\in\Sigma$ by replacing all occurrences of $X_i$ with $a$.
	Subsequently, we apply the reduction rule~R-1 of \cref{secGrammarIrreducibility}.
\end{description}

We collected the obtained grammar sizes in \Tabref{mergedtable2}.
There, we observe that the Lyndon SLP is no match for Re-Pair, but
competitive with LCA and Recompression.
Although this evaluation puts Re-Pair in a good light, it seems hard to build 
an index data structure on this grammar that can be as efficient as 
the self-index data structure based on the Lyndon SLP, which we present in the next section.

\begin{table*}
  \small\centering
  \caption{Sizes of the resulting grammars benchmarked in \cref{secExperiments}.}
    \label{mergedtable2}
    
\begin{tabular}{l|c*{4}{r}}
    \hline \hline
    \multicolumn{1}{c|}{collection} &
    \multicolumn{1}{c}{} &
    \multicolumn{1}{c}{Re-Pair} &
    \multicolumn{1}{c}{LCA} &
    \multicolumn{1}{c}{Recompression} &
    \multicolumn{1}{c}{Lyndon SLP} \\
    \hline
    \texttt{cere}
    & SLP & 6,433,183 & 9,931,777 & 8,537,747 & 13,026,562 \\
    & AG  & 4,057,693 & 6,513,345 & 5,309,789 & 7,469,979 \\
    \hline
    \texttt{einstein.de.txt}
    & SLP & 125,343 & 251,411 & 202,749 & 205,348 \\
    & AG  & 84,493 & 168,193 & 127,790 & 123,963 \\
    \hline
    \texttt{kernel}
    & SLP & 2,254,840 & 4,065,522 & 3,587,382 & 4,201,895 \\
    & AG  & 1,373,244 & 2,507,291 & 2,135,779 & 2,400,211 \\
    \hline
    \texttt{world\_leaders}
    & SLP & 601,757 & 1,243,757 & 1,023,739 & 911,222 \\
    & AG  & 398,234 & 809,163 & 636,700 & 552,497 \\
    \hline
\end{tabular}

\end{table*}

\section{Lyndon SLP based self-index}\label{secSelfIndex}

Given a Lyndon SLP of size $g$, we can build an indexing data structure on it
to query all occurrences of a pattern~$P$ of length~$m \in [1..n]$ in~$T$.
We call this query~$\Locate(P)$.
Our data structure is based on the approach of~\Cite{claudear:_self_index_gramm_based_compr}.
This approach separates the occurrences of a pattern into so-called primary occurrences and secondary occurrences.
It first locates the primary occurrences and, with the help of these, it subsequently locates the secondary occurrences.
To this end, it locates primary occurrences with a labeled binary relation data structure, and subsequently locates the secondary occurrences with the grammar tree.
In our case, we find the primary occurrences with so-called partition pairs.

\begin{figure}
	\begin{center}
		\includegraphics[width=0.4\textwidth]{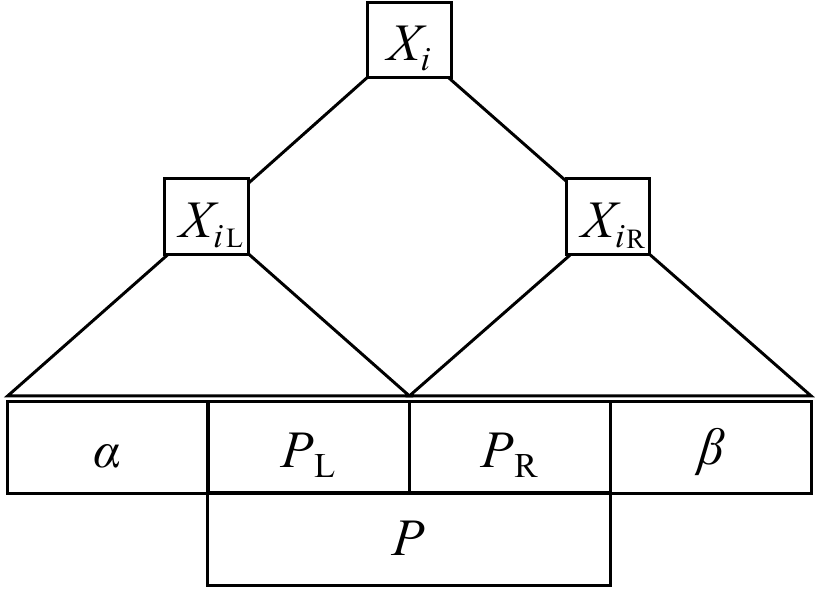}
	\end{center}
	\caption{%
		A partition pair~$(\LeftPattern, \RightPattern)$ of a pattern~$P$ with one of its associated tuples~$(X_i, \alpha, \beta)$.
	}
	\label{fig:PartitionPair}
\end{figure}

A \teigi{partition pair} (at position~$i$) of a pattern~$P[1..m]$ is a pair $(P[1..i],P[i+1..m])$ with $i \in [1..m]$ such that there exists
a rule $X_i \rightarrow \LeftRule \RightRule$ with $\derive(\LeftRule)$ and $\derive(\RightRule)$ having
$P[1..i]$ and $P[i+1..m]$ as a (not necessarily proper) suffix and as a prefix, respectively.
Similar to the grammar proposed in \myCite{Sect.~6.1}{christiansen18optimaltime}, we can bound the number of partition pairs by \Oh{\lg m}
by carefully selecting all possible partition pairs:

Given a partition pair~$(\LeftPattern, \RightPattern)$ of $P$,
let $X_i \rightarrow \LeftRule \RightRule$ be a rule such that
$\derive(\LeftRule)$ and $\derive(\RightRule)$ have $\LeftPattern$ and $\RightPattern$ as a suffix and as a prefix, respectively.
Consequently, there exist two strings~$\alpha$ and~$\beta$ such that
$\derive(\LeftRule) = \alpha \LeftPattern$ and $\derive(\RightRule) = \RightPattern \beta$ (cf.~\figref{fig:PartitionPair}).
By the definition of the Lyndon tree of the text~$T$,
$(\derive(\LeftRule), \derive(\RightRule)) = (\alpha \LeftPattern, \RightPattern \beta)$ is the standard factorization of
$\derive(X_i) = \alpha \LeftPattern \RightPattern \beta$.
According to the standard factorization, $\RightPattern \beta$ is the longest suffix of $\derive(X_i)$ that is a Lyndon word.
For the proofs of \cref{lemLyndonWordRuns,lemSignificantSuffix}, we use this notation and call the tuple
$(X_i, \alpha, \beta)$ a \teigi{tuple associated with} $(\LeftPattern, \RightPattern)$.

Let us take $P := \texttt{bab}$ as an example.
The only partition pair is $(\texttt{b}, \texttt{ab})$.
Considering the Lyndon grammar of our example text given in \figref{fig:LyndonSLP}, the tuples associated with $(\texttt{b}, \texttt{ab})$ are
$(X_8, \texttt{aa}, \varepsilon)$ and $(X_5, \texttt{a}, \texttt{b})$.

Note that $|\alpha| = 0$ if $P$ is a Lyndon word. If $P$ is a proper prefix of a Lyndon word\footnote{I.e., there is a string~$S \in \Sigma^+$ such that $PS$ is a Lyndon word.}, then $\alpha$ may be empty.
If $P$ is a not a (not necessarily proper) prefix of a Lyndon word, then $|\alpha| > 0$ (since $\alpha\LeftPattern\RightPattern\beta$ is a Lyndon word).

\subsection{Associated tuples with non-empty $\alpha$}
We want to reduce the number of possible partition pairs from $m$ to \Oh{\lg m}.
A first idea is that only the beginning positions of the Lyndon factors of~$P$ contribute to potentially partition pairs.
We prove this in \cref{lemPriOccLyndonWord}, after defining the Lyndon factors:

The (composed) \teigi{Lyndon factorization}~\cite{ChenFL58:_lyndon_factorization_} of a string~$P\in\Sigma^+$
is the factorization of $P$ into a sequence~$P_1^{\tau_1} \cdots P_p^{\tau_p}$ of lexicographically decreasing Lyndon words~$P_1, \ldots, P_p$,
where (a) each $P_x\in\Sigma^+$ is a Lyndon word, and (b) $P_x \succ P_{x+1}$ for each $x \in [1..p)$.
$P_x$ and $P_x^{\tau_x}$ are called \teigi{Lyndon factor} and \teigi{composed Lyndon factor}, respectively.

\begin{lemma}[{\myCite{Algo.\ 2.1}{duval83lyndon}}]\label{lem:LyndonFactorizationLinearTime}
	The Lyndon-factorization of a string can be computed in linear time.
\end{lemma}
We borrow from \myCite{Sect.\ 2.2}{i16faster} the notation $\lfs{P}{x} := P_x^{\tau_x} \cdots P_p^{\tau_p}$ for the suffix of $P$ starting with the $x$-th Lyndon factor.
Given $\lambda_P \in [1..p]$ is the smallest integer such that $\lfs{P}{x+1}$ is a prefix of $P_x$ for every $x \in [\lambda_P..p-1]$,
$\lfs{P}{x}$ is called a \teigi{significant suffix} of~$P$ for every $x \in [\lambda_P..p]$.
Consequently, $\lfs{P}{p} = P_p^{\tau_p}$ is a significant suffix.

In what follows, we show that $\RightPattern$ of a partition pair~$(\LeftPattern,\RightPattern)$ has to start with
a Lyndon factor (\cref{lemPriOccLyndonWord}), and further has to start with a composed Lyndon factor (\cref{lemLyndonWordRuns}).
Finally, we refine this result by restricting $\RightPattern$ to begin with a significant suffix (\cref{lemSignificantSuffix})
whose number is bounded by the following lemma:

\begin{lemma}[{\myCite{Lemma 12}{i16faster}}]\label{lem:LambdaUpperBound}
	The number of significant suffixes of $P$ is \Oh{\lg m}.
\end{lemma}

In what follows, we study the occurrences of~$P$ in~$T$ under the circumstances that $T$ is represented by its Lyndon tree induced by the standard factorization,
while~$P$ is represented by its Lyndon factors.

\begin{lemma}[{\myCite{Prop.~1.10}{duval83lyndon}}]\label{lemDuvalLongestPrefixLyndon}
	The longest prefix of $P$ that is a Lyndon word is the first Lyndon factor~$P_1$ of~$P$.
\end{lemma}

\begin{lemma}[{\myCite{Lemma~{5.4}}{bannai17runs}}]\label{lemLyndonSubstring}
	Given a production $X_j \rightarrow \LeftRule[j] \RightRule[j] \in \LYN{}$,
	there is no Lyndon word that is a substring of $\derive(X_j) = \derive(\LeftRule[j])\derive(\RightRule[j])$
	beginning in $\derive(\LeftRule[j])$ and ending in $\derive(\RightRule[j])$,
	except $\derive(X_j)$.
\end{lemma}

\begin{lemma}\label{lemPriOccLyndonWord}
	Given $(\LeftPattern, \RightPattern)$ is a partition pair of a pattern $P$,
	$\RightPattern$ starts with a Lyndon factor of $P$ if
	there is an associated tuple~$(X_i, \alpha, \beta)$ with $|\alpha| > 0$.
\end{lemma}
\begin{proof}
	Since $|\alpha| > 0$ holds, $P$ is a proper substring of $\derive(X_i)$.
	Then $\RightPattern$ must start with a Lyndon factor of $P$ according to \cref{lemLyndonSubstring}.
\end{proof}

\begin{lemma}[{\myCite{Prop.~1.3}{duval83lyndon}}]\label{lemLyndonConcat}
	Given two Lyndon words~$\alpha$, $\beta$ with $\alpha \prec \beta$,
	the concatenation $\alpha \beta$ is also a Lyndon word.
\end{lemma}

\begin{lemma}\label{lemLyndonWordRuns}
	Given $(\LeftPattern, \RightPattern)$ is a partition pair of a pattern $P$,
	$\RightPattern$ starts with a composed Lyndon factor of $P$ if
	there is an associated tuple~$(X_i, \alpha, \beta)$ with $|\alpha| > 0$.
\end{lemma}
\begin{proof}
	Let $(X_i, \alpha, \beta)$ be a tuple associated with $(\LeftPattern, \RightPattern)$.
	Assume for the contrary that $\RightPattern$ does not start with any composed Lyndon factors of $P$,
	namely, there exists $x \in [1..p]$ and $k \in [1..\tau_x-1]$
	such that $\LeftPattern$ and $\RightPattern$ have
	$P_x^{\tau_x-k}$ and $P_x^{k}$ as a suffix and prefix, respectively (cf.~\figref{fig:figLyndonWordRuns}).
	By the assumption, $\derive(\RightRule) = P_x^{k} \lfs{P}{x+1} \beta$ is the longest Lyndon word
	that is a suffix of $\derive(X_i)$.
	Since $P_x \prec \derive(\RightRule)$ and $P_x$ is a Lyndon word,
	$P_x \derive(\RightRule)$ is also a Lyndon word by \cref{lemLyndonConcat}.
	This contradicts that $\derive(\RightRule)$ is the longest Lyndon word
	that is a suffix of $\derive(X_i)$.
\end{proof}

\begin{figure}
	\begin{center}
		\includegraphics[width=0.6\textwidth]{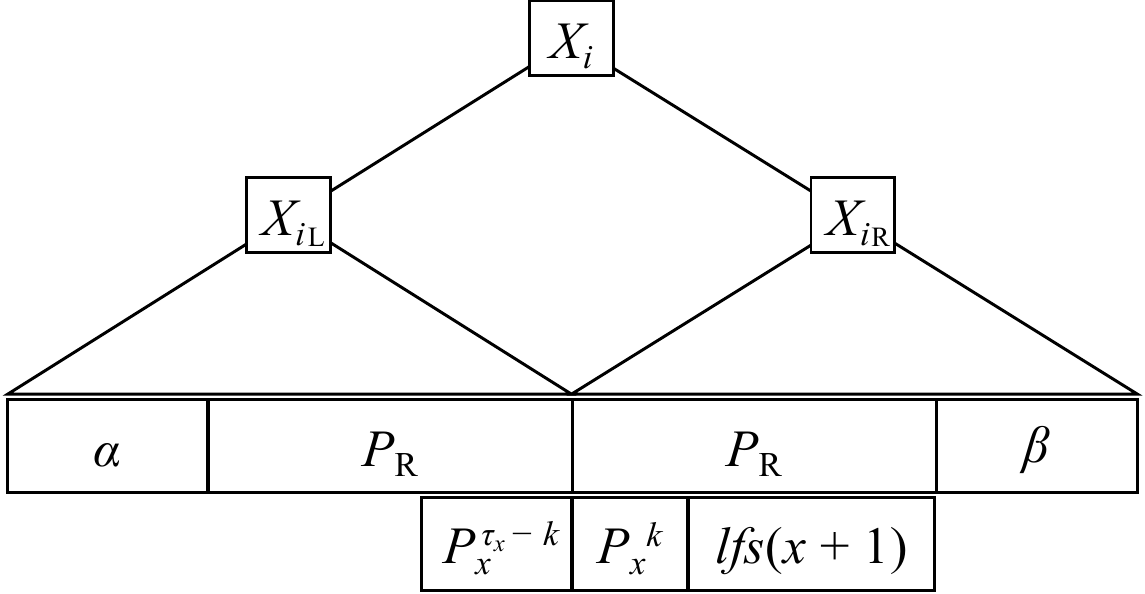}
	\end{center}
	\caption{%
		Setting of the proof of \cref{lemLyndonWordRuns}.
	}
	\label{fig:figLyndonWordRuns}
\end{figure}

\cref{lemLyndonWordRuns} helps us to concentrate on the \emph{composed} Lyndon factors.
Next, we show that only those composed Lyndon factors are interesting that start with a significant suffix:

\begin{lemma}\label{lemSignificantSuffix}
	Given $(\LeftPattern, \RightPattern)$ is a partition pair of a pattern $P$,
	then $\RightPattern$ is a significant suffix of $P$ if
	there is an associated tuple~$(X_i, \alpha, \beta)$ with $|\alpha| > 0$.
\end{lemma}
\begin{proof}
	Let $(X_i, \alpha, \beta)$ be a tuple associated with $(\LeftPattern, \RightPattern)$ and $|\alpha| > 0$.
	By \cref{lemLyndonWordRuns},
	there exists $x \in [1..p]$ such that $\RightPattern = P_x^{\tau_x} \cdots P_p^{\tau_p}$.
	Assume for the contrary that $x < \lambda_P$, i.e., $\RightPattern$ is not a significant suffix of~$P$.
	By definition, $\lfs{P}{x} \succ \lfs{P}{x+1}$ holds.
	Since $\lfs{P}{x+1}$ is not a prefix of $\lfs{P}{x}$,
	$\lfs{P}{x} \beta \succ \lfs{P}{x+1} \beta$ also holds.
	This implies that $\RightPattern = \lfs{P}{x} \beta$ is not a Lyndon word, a contradiction.
\end{proof}
This, together with \cref{lem:LambdaUpperBound}, yields the following corollary.

\begin{corollary}\label{corPartitionPairsNonEmpty}
	There are \Oh{\lg m} partition pairs of~$P$ associated with a tuple~$(X_i, \alpha, \beta)$ with $|\alpha| > 1$.
\end{corollary}

Let us take $P := \texttt{abacabadabacababa}$ as an elaborated example.
Its composed Lyndon factorization is $P = P_1 P_2 P_3^2 P_4$, where its Lyndon factors are
$P_1 = \texttt{abacabad}$, $P_2 = \texttt{abac}$, $P_3 = \texttt{ab}$, and $P_4 = \texttt{a}$ with $\lambda_P = 3$.
Hence, $\lfs{P}{3}$ and $\lfs{P}{4}$ are significant suffixes.
Its potential partition pairs are $(P_1 P_2, P_3^2 P_4)$, $(P_1 P_2 P_3^2, P_4)$.
There is no Lyndon SLP such that another partitioning like $(P_1, P_2 P_3^2 P_4)$ or $(P_1 P_2 P_3, P_3 P_4)$ would have an associated tuple
according to \cref{lemSignificantSuffix} and \cref{lemLyndonWordRuns}, respectively.

\subsection{Associated tuples with empty $\alpha$}\label{secAlphaEmpty}
Given a partition pair~$(\LeftPattern, \RightPattern)$ associated with a tuple~$(X_i, \varepsilon, \beta)$,
we consider two cases depending on $|\LeftPattern|$:
In the case of $|\LeftPattern| = 1$, $(P[1], P[2..m])$ may be a partition pair of $P$.
In the case of $|\LeftPattern| \ge 2$, suppose that $P^{\prime} = P[2..m]$, $\alpha^{\prime} = P[1]$
and $(\LeftPattern^{\prime}, \RightPattern^{\prime})$ is a partition pair of $P^{\prime}$
with associated tuple~$(X_i, \alpha^{\prime}, \beta)$.
Then, $(P[1]\LeftPattern^{\prime}, \RightPattern^{\prime})$ is a partition pair of $P$ with associated tuple~$(X_i, \varepsilon, \beta)$.
We can use \cref{lemPriOccLyndonWord}, \cref{lemLyndonWordRuns} and \cref{lemSignificantSuffix}
to restrict $\RightPattern^{\prime}$ starting with a significant suffix of $P[2..m]$ (cf.~\cref{corPartitionPairsNonEmpty}).

\begin{corollary}\label{corPartitionPairsEmpty}
	There are $\Oh{\lg m}$ partition pairs of~$P$ associated with a tuple~$(X_i, \varepsilon, \beta)$.
\end{corollary}

Combining \cref{corPartitionPairsNonEmpty} with \cref{corPartitionPairsEmpty} yields the following theorem and the main result of this subsection:

\begin{theorem}\label{corPartitionPairs}
	There are \Oh{\lg m} partition pairs of a pattern of length~$m$.
\end{theorem}

\subsection{Locating a pattern}

In the following, we use the partition pairs to find all primary occurrences.
We do this analogously as for the $\Gamma$-tree~(\myCite{Sect.~3.1.}{navarro19indexing})
or for special grammars~(\myCite{Sect.~6.1}{christiansen18optimaltime}).

\begin{lemma}[{\myCite{Lemma 5.2}{gagie18bwt}}]\label{lemSuffixRangeQuery}
	Let $\mathcal{S}$ be a set of strings and assume that we can
	(a) extract a substring of length~$\ell$ of a string in $\mathcal{S}$ in time~$f_e(\ell)$
	and (b) compute the Karp-Rabin fingerprint~\cite{DBLP:journals/ibmrd/KarpR87} of a substring of a string in $\mathcal{S}$ in time~$f_h$.
	Then we can build a data structure of \Oh{|\mathcal{S}|} words solving
	the following problem in $\Oh{m \lg\sigma/w + t(f_h + \lg m)+ f_e(m)}$ time:
		given a pattern $P[1..m]$ and $t > 0$ suffixes $Q_1, \ldots , Q_t$ of $P$,
	discover the ranges of strings in (the lexicographically-sorted) $\mathcal{S}$ prefixed by $Q_1, \ldots, Q_t$.
\end{lemma}

\begin{lemma}[{\myCite{Thm.~1.1}{bille15randomaccess}}]\label{lemSubstringExtraction}
	For an AG of size $g$ representing a string of length $n$ we can
	extract a substring of length $\ell$ in time $\Oh{\ell + \lg n}$
	after \Oh{g} preprocessing time and space.
\end{lemma}

\begin{lemma}[{\myCite{Thm.~1}{bille17fingerprints}}]\label{lemFingerprintComputation}
	Given a string of length $n$ represented by an SLP of size $g$,
	we can construct a data structure supporting fingerprint queries in
\Oh{g} space and $\Oh{\lg n}$ deterministic query time.
This data structure can be constructed in \Oh{n \lg n} randomized time~(cf.~\myCite{Sect.~2.4}{gagie14lz77})
by using Karp, Miller and Rosenberg's~\cite{karp72rapid} renaming algorithm to make all fingerprints unique.
\end{lemma}

With \cref{lemSubstringExtraction} and \cref{lemFingerprintComputation} we have
$f_e(\ell) = \Oh{\ell + \lg n}$ and $f_h = \Oh{\lg n}$ in \cref{lemSuffixRangeQuery}, respectively, leading to:

\begin{corollary}\label{corSuffixRangeQuery}
	There is a data structure using \Oh{g} space such that,
	given a pattern $P[1..m]$ with $m \le n$, it can find all variables whose derived strings have one of $t$ selected suffix of $P$ as a prefix
	in $\Oh{m \lg \sigma / w + t(\lg n + \lg m) + \ell + \lg n}$ time.
\end{corollary}
\Cref{corSuffixRangeQuery} yields $\Oh{m \lg n}$ time for $t = m$, i.e., when we need to split the pattern at each position.
It yields $\Oh{m \lg \sigma / w + \lg m \lg n}$ time for $t = \lg m$,
i.e., the case for \myCite{Sect.~6.1}{christiansen18optimaltime} and for the Lyndon SLP thanks to \cref{lemSignificantSuffix}
(we assume that the pattern is not longer than the text).

We can retrieve the associated tuples of all primary occurrences by plugging the variables retrieved in \cref{corSuffixRangeQuery}
into a data structure for labeled binary relations~\cite{claudear:_self_index_gramm_based_compr}.

For that, we generate two list $\mathcal{L}$ and $\mathcal{L}^\mathrm{REV}$ of all variables~$X_1,\cdots,X_g$ of the grammar sorted lexicographically by their derived strings and the reverses of their derived strings, respectively.
Both lists allow us to answer a prefix (resp.\ suffix) query by returning a range of variables having the prefix (resp.\ suffix) in question.
The query is performed by the data structure described in \cref{lemSuffixRangeQuery} (with $\mathcal{S}$ being either $\mathcal{L}$ or $\mathcal{L}^\mathrm{REV}$).
Finally, we can plug the obtained ranges into the labeled binary relation data structure of Claude and Navarro~\cite{claudear:_self_index_gramm_based_compr}:

\begin{lemma}[{\myCite{Thm.~3.1}{claudear:_self_index_gramm_based_compr}}]\label{lemFindAssociatedTuple}
	Given two list $\mathcal{L}$ and $\mathcal{L}^\mathrm{REV}$ of variables sorted lexicographically by their expressions and its reversed strings,
	we can built a data structure of $\Oh{g}$ words of space in $\Oh{g \lg g}$ time for supporting the following query:
	Given a partition pair~$(\LeftPattern,\RightPattern)$ and ranges in~$\mathcal{L}$ and $\mathcal{L}^\mathrm{REV}$ of those variables
	whose derived strings have $\derive(\RightPattern)$ as a prefix and $\derive(\LeftPattern)$ as a suffix,
	this data structure can retrieve all associated tuples of~$(\LeftPattern,\RightPattern)$
	in $\Oh{(1 + {\occ[\LeftPattern,\RightPattern]^\prime}) \lg g}$ time,
	where $\occ[\LeftPattern,\RightPattern]^\prime$ denotes their number.
\end{lemma}

The time complexity of \cref{corSuffixRangeQuery} and \cref{lemSuffixRangeQuery} is based on the assumption that
we have (static) z-fast tries~\cite{belazzougui10prefix} built on the lists~$\mathcal{L}$ and $\mathcal{L}^\mathrm{REV}$,\footnote{We use again the derived string or, respectively, the reverse of the derived string of each non-terminal in one of the lists as its respective keyword to insert into the trie.}
which we can build in \Oh{g} expected time and space~(\myCite{Sect.~6.6 (3)}{christiansen18optimaltime}).

Since there are \Oh{\lg m} partition pairs according to \cref{corPartitionPairs},
applying \cref{lemFindAssociatedTuple} over all \Oh{\lg m} partition pairs yields $\Oh{\lg m \lg g + {\occ[prim]} \lg g}$ time,
where $\occ[prim]$ denotes the number of all primary occurrences.

\begin{corollary} \label{corFindAssociatedTuple}
	We can find the primary occurrences of a pattern~$P$ in
	\[
	\underbrace{\Oh{m}}_{\cref{lem:LyndonFactorizationLinearTime}} + \underbrace{\Oh{m \lg \sigma / w + \lg m \lg n}}_{\cref{corSuffixRangeQuery}} + \underbrace{\Oh{\lg m \lg g + \occ \lg g}}_{\cref{lemFindAssociatedTuple}}  =
		\]
			$= \Oh{m + \lg m \lg n + \occ \lg g}$ time.
\end{corollary}

Finally, we use the derivation tree to find the remaining (secondary) occurrences of the pattern:

\subsection{Search for secondary occurrences}\label{secSecondaryOcc}
We follow Claude and Navarro~\cite{DBLP:conf/spire/ClaudeN12a} improving the search of the secondary occurrences in \Cite{claudear:_self_index_gramm_based_compr}
by applying reduction rule R-1 to enforce C-1
(see Sect.~\ref{secGrammarIrreducibility}).
The resulting admissible grammar~$\CFG$ is no longer an SLP in general.
Since we only remove variables with a single occurrence,
the size of $\CFG$ is \Oh{g}.
Consequently, we can store both \CFG{} \emph{and} \SLP{} in \Oh{g} space.

\begin{lemma}[{\myCite{Sect.~5.2}{DBLP:conf/spire/ClaudeN12a}}] \label{lemRangeQuery}
Given the associated tuples of all partition pairs, we can find all $\occ$ occurrences of~$P$ in~$T$
	with \CFG{} in \Oh{\occ \lg g} time.
\end{lemma}
Remembering that we split the analysis
of an associated tuple in the cases $|\alpha| = 0$ (\cref{secAlphaEmpty}) and $|\alpha| > 0$ (\cref{corFindAssociatedTuple}),
we observe that the time complexity of the latter case dominates.
Combining this time with \cref{lemRangeQuery} yields the time complexity for answering $\Locate(P)$ with the Lyndon SLP:

\begin{theorem}
	\label{th:collision_free}
	Given the Lyndon SLP of~$T$,
	there is a data structure using $\Oh{g}$ words that can be constructed in $\Oh{n \lg n}$ expected time,
	supporting $\Locate(P)$ in $\Oh{m + \lg m \lg n + \occ \lg g}$ time for a pattern~$P$ of length~$m$.
\end{theorem}
Note that the \Oh{n \lg n} expected construction time is due to the data structure described in \cref{lemFingerprintComputation}.

\section{Conclusion}

We introduced a new class of SLPs, named the \emph{Lyndon SLP}, and
proposed a self-index structure of $\Oh{g}$ words of space,
which can be built from an input string $T$ in $\Oh{n \lg n}$ expected time,
where $n$ is the length of~$T$ and
$g$ is the size of the Lyndon SLP for $T$.
By exploiting combinatorial properties on Lyndon SLPs,
we showed that
$\Locate(P)$ can be computed
in $\Oh{m + \lg m \lg n + \occ \lg g}$ time for a pattern~$P$ of length~$m$, where $\occ$ is the number of occurrences of $P$.
This is better than the
${\Oh{m^2 \lg \lg_{\ag} n + (m + \occ)\lg\ag}}$
query time of the SLP-index by Claude and Navarro~\cite{DBLP:conf/spire/ClaudeN12a} (cf.~\Tabref{tabConstruction}),
which works for a \emph{general} admissible grammar of size~$\ag$.

We have not implemented the proposed self-index structure,
and comparing it with other self-index implementations
such as
    the FM index~\cite{ferragina08compressed},
    the LZ index~\cite{arroyuelo11lz78},
    the ESP index~\cite{DBLP:conf/wea/TakabatakeTS14}, or
    the LZ-end index~\cite{kreft13lzend}
will be a future work.
Also, we want to speed up the query time to $\Oh{m \lg \sigma / w + \lg m \lg n + \occ \lg g}$ by applying broadword techniques for determining the Lyndon factors of the pattern~$P$ (cf.~\cref{corFindAssociatedTuple}),
where
$\sigma$ is the alphabet size and $w$ is the computer word length.

\bibliography{ref}
\bibliographystyle{abbrv}
\end{document}